\documentclass[sigplan,10pt]{acmart}
\renewcommand\footnotetextcopyrightpermission[1]{}

\AtBeginDocument{}

\usepackage{amsmath}
\usepackage{filecontents}
\usepackage{times}  
\usepackage{graphicx}
\usepackage{xspace}
\usepackage{xcolor}
\usepackage[T1]{fontenc}
\usepackage{caption}
\usepackage{subcaption}
\usepackage{verbatim}
\usepackage{url}
\usepackage{enumitem}
\usepackage{listings}
\usepackage{multirow}
\usepackage{pifont}
\usepackage{mathtools}
\usepackage{fancyhdr}
\usepackage{xurl}
\usepackage{textcomp}
\usepackage{pifont}
\usepackage{afterpage}
\usepackage[export]{adjustbox}
\usepackage{tabularx}
\usepackage{algorithm}
\usepackage{algpseudocode}
\usepackage{tikz}
\usepackage{placeins}
\usetikzlibrary{arrows.meta, positioning, shapes.geometric, calc}

\usepackage[capitalise,noabbrev]{cleveref}
\crefformat{section}{\S#2#1#3}
\crefname{appendix}{Appendix}{Appendices}

\usepackage{amsmath, amsthm}
\newtheorem{definition}{Definition}
\newtheorem{theorem}{Theorem}

\newcommand{\parab}[1]{\noindent\textbf{#1}}
\newcommand{\sysname}{InstantInfer\xspace}

\definecolor{mypurple}{RGB}{102,0,153}

\begin{document}
\pagestyle{plain}

\title{\sysname: Enabling Fast LLM Cold Start with Communicating Finite Automata}

\settopmatter{authorsperrow=4}

\author{Yitao Yuan}
\authornote{This work was done during an internship at ScitiX AI.}
\affiliation{\institution{Peking University}
  \city{}
  \country{}}
\affiliation{\institution{ScitiX AI}
  \city{}
  \country{}}

\author{Yongchao He}
\authornote{Corresponding authors.}
\affiliation{\institution{ScitiX AI}
  \city{}
  \country{}}

\author{Shaoke Fang}
\affiliation{\institution{Peking University}
  \city{}
  \country{}}

\author{Wenfei Wu}
\authornotemark[2]
\affiliation{\institution{Peking University}
  \city{}
  \country{}}

\renewcommand{\shortauthors}{Yuan et al.}

\begin{abstract}

Cold starts in large language model (LLM) inference services significantly affect user experience, yet they remain inefficient due to sequential initialization and a massive number of fine-grained I/O requests issued by complex software components. Although refactoring the program can yield advantages such as concurrent execution and I/O merging, this approach is error-prone and carries correctness risks when dealing with massive, heterogeneous components. We propose the Communicating Finite Automata (CFA) abstraction to systematically analyze cross-component optimization opportunities, and design a programming framework to enable CFA-based component program refactoring. This framework preserves the original sequential program structure while enabling safe concurrent component execution. We prove the correctness of the program refactoring. We apply the CFA abstraction and framework to refactor process tree creation, tensor loading, and model switching in vLLM, forming a new cold-start system named \sysname. Extensive experiments demonstrate that \sysname substantially accelerates LLM cold starts (achieving up to $7.2\times$ speedup) and exhibits robustness across diverse GPUs, workloads, and scales.

\end{abstract}

\maketitle

\section{Introduction}\label{sec:introduction}

Large language model (LLM) inference has become a critical service primitive for various applications~\cite{github-copilot,microsoft-365-copilot,claude-code,codex,openclaw}.
Serving these interactive applications at scale requires expensive GPU resources~\cite{fu2024serverlessllm}, while providers must support diverse models and bursty user demand.
To improve GPU utilization, inference service providers share GPU clusters across many models, keep only a subset active, and activate
other models on demand, as in typical serverless or elastic serving scenarios~\cite{fu2024serverlessllm,liu2025pipeboost,zhang2025blitzscale}.
This on-demand activation introduces LLM cold start: before serving requests, the system must load the model from disk to the GPU and establish a ready state for inference.
Existing works~\cite{fu2024serverlessllm,zeng2025medusa,liu2025pipeboost,zhang2025blitzscale} and our measurements show that cold-start latency dominates the time to first token (TTFT): the startup time is around 50s for <10B models, around 160s for \textasciitilde70B models, and around 300s for >400B models, accounting for 99.6\%--99.9\% of the total TTFT.
Although warm starts can be enabled by pre-placing models in memory, cold starts remain necessary during bursty access to massive long-tail models or unpredictable crash recovery, still degrading user experience (\cref{ssec:problems}).

The entire cold-start process runs in multiple stages, with each stage initializing massive, heterogeneous, and possibly hierarchical components.
Existing works on LLM cold starts mostly accelerate a particular stage of initialization, such as model loading~\cite{fu2024serverlessllm,lou2025hydraserveminimizingcoldstart,yu2025lambda,liu2025pipeboost,sui2024pre,yoshimura2025speeding,runai-model-streamer} or serving-readiness optimization~\cite{zeng2025medusa,liu2026foundrytemplatebasedcudagraph}, usually through techniques such as caching~\cite{zhang2025blitzscale,zhu2025tangram,sui2024pre} and state reuse~\cite{zhu2025tangram,zeng2025medusa}.
In this work, we do not optimize a single stage or technique in isolation, but focus on the dependency structure of the initialization process itself and use a unified abstraction to reveal latent optimization opportunities through cross-component refactoring.

Having observed two significant overheads during LLM cold start, we reveal opportunities to perform \textbf{\textit{cross-component optimization}} to accelerate overall performance.
First, for ease of programming and debugging, numerous components are typically executed sequentially along the program control flow, which forces unnecessary sequential dependencies and waiting.
Second, component operations originate from high-level algorithms, but their granularity (e.g., tensor size) mismatches the efficient operational granularity of the underlying hardware (e.g., I/O sizes), leading to high overhead or low hardware utilization (\cref{ssec:analysis}).

Our intuition is to \textbf{refactor the startup program} to exploit two opportunities: overlapping execution times across components and merging fine-grained I/O requests to improve hardware throughput, and ultimately reduce the overall initialization time.
However, the LLM startup program exhibits a complex structure, which poses significant challenges for program refactoring.
The cold start spans multiple stages, including process initialization, model initialization, and optimization; across these stages, there are numerous and heterogeneous components, organized in a hierarchy or sequence (e.g., process tree and model tensors).
Directly rewriting the existing code of components is not only tedious and labor-intensive, but worse yet, it inherently risks concurrency hazards.
To enable low-effort, correct, and safe refactoring of cold-start programs, we overcome three challenges.

First, existing program components are viewed as ``monolithic'' entities, discouraging fine-grained internal state analysis.
The program's control flow runs the monolithic components sequentially, so directly launching them in parallel could lead to race conditions and hazards.
We propose \textbf{a \textit{Communicating Finite Automata (CFA)} abstraction} to systematically explore the optimization opportunities.
Within a component, execution progress is described by monotonic state transitions; across components, logical dependencies are described by component-state dependencies.
The CFA abstraction enables direct application of the two optimizations: independent state transitions across components can safely run concurrently, and many fine-grained FAs/state transitions can be merged into a coarser FA/transition to better match the hardware data operation granularity (\cref{ssec:model}).
 
Second, an appropriate programming framework is needed to reduce the development effort of program refactoring.
Components exist across multiple stages and levels in the current startup program; modifying the heterogeneous components is tedious and error-prone, and managing all components' FAs in a single space would lead to state explosion.
We propose a \textbf{{CFA-based programming framework}} with unified declarative interfaces for all components.
The programming framework inserts statements into the component program to declare state transitions and state dependencies, requiring minimal changes to the original program.
The runtime performs dependency checking and blocks or wakes components accordingly.
The framework also supports managing related FAs in an isolated namespace, avoiding excessive inter-FA messaging (\cref{ssec:programming}).

Third, the program refactoring must guarantee correctness.
Overlapping component executions introduce concurrency, which could lead to race conditions and hazards on shared data or states.
We provide a \textbf{rigorous proof} that sequential monolithic component execution is \textbf{equivalent} to CFA-refactored execution under the declared state dependencies.
We also analyze data safety in each case where we use CFA to refactor the LLM startup program (\cref{ssec:correctness}).

\begin{figure}[t]
    \centering
    \includegraphics[width=\linewidth]{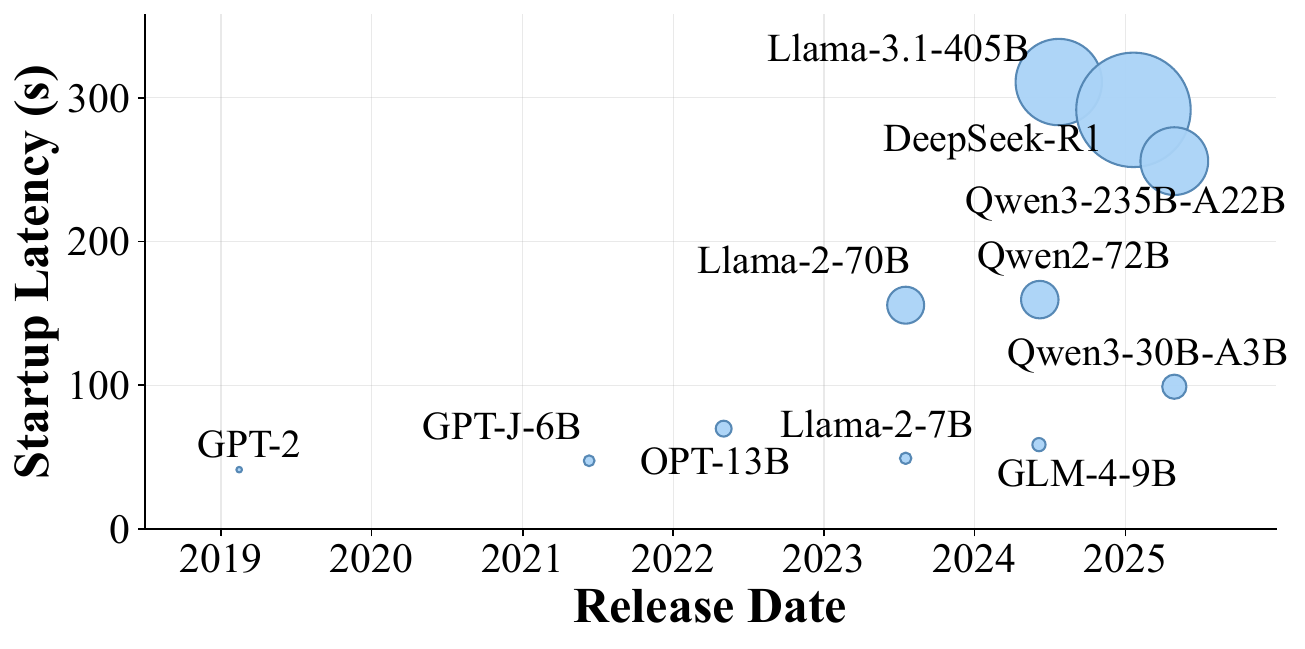}
    \caption{Startup latency vs model size. The circle area is proportional to the model parameter count.}
    \label{fig:startup_latency_vs_release}
\end{figure}

We instantiate the CFA abstraction and programming framework in three representative cold-start cases (\cref{sec:cfa-for-llm} and \cref{sec:implementation}): process-tree materialization, tensor loading, and model switching (cold start with replacement).
The related implementation has been integrated into vLLM, forming a new LLM cold-start system \sysname.
We run extensive experiments and show that \sysname achieves up to 7.2$\times$ lower cold-start TTFT, up to 32.3$\times$ faster model loading, and up to 11.8$\times$ lower service stall during model switching across diverse hardware and models (\cref{sec:evaluation}).

In summary, this paper makes the following contributions:
\begin{itemize}[leftmargin=*]
    \item We identify cross-component optimization opportunities in LLM cold start, including overlapping component execution and merging fine-grained I/O operations.
    \item We propose the CFA abstraction to describe cold-start components as monotonic state transitions with explicit state dependencies, enabling systematic program refactoring.
    \item We design a CFA-based programming framework and implement \sysname in vLLM to refactor process-tree materialization, tensor loading, and model switching.
    \item We evaluate \sysname across diverse hardware and models, demonstrating its effectiveness in accelerating LLM cold start.
\end{itemize}

\section{Background}\label{sec:background}

\subsection{Problems in LLM Cold Start} \label{ssec:problems}
Cold start is highly prevalent in practical deployments. (i) LLM inference providers may maintain numerous models and load specific models on demand in response to user requests for GPU resource efficiency~\cite{duan2024muxserve,yu2025prism,xiang2025aegaeon}. (ii) When a sudden burst of requests for a specific model exceeds the capacity of active instances, cold starts are needed to launch new instances for elastic scaling~\cite{fu2024serverlessllm,liu2025pipeboost,zhang2025blitzscale}. (iii) For LLM service developers, updating model weights or modifying deployment configurations, such as adjusting parallelism, inevitably requires reloading the entire model~\cite{kwon2023vllm,zheng2024sglang}. (iv) When unpredictable disasters or system failures crash active LLM instances, the system must restart instances via cold start to recover the service~\cite{fu2024serverlessllm,liu2025pipeboost}.

Although some existing works demonstrate that large models can be prepared in host memory to enable faster loading to GPU~\cite{yu2025torporgpuenabledserverlesscomputing}, models cannot always be kept ready in memory under all circumstances~\cite{fu2024serverlessllm,zhang2025blitzscale,hong2024optimus}: host memory may be insufficient, and predictive in-memory caching cannot always guarantee cache hits, which may waste host memory. 

We measured the startup time of several popular models and show the results in \cref{fig:startup_latency_vs_release}\footnote{Llama-3.1-405B has a larger storage footprint and longer loading time than DeepSeek-R1 because the former uses 16-bit weights and the latter 8-bit.}. The startup time is around 50s for <10B models, and around 160s and 300s for \textasciitilde70B and >400B models, respectively. The typical time to first token (TTFT) for a request of length 1000 when the model is ready ranges from 0.08--0.81s. Cold start accounts for the overwhelming majority (99.6\%--99.9\%) of total latency. This pattern has also been observed in other studies~\cite{fu2024serverlessllm,zeng2025medusa}.

\subsection{Problem Analysis} \label{ssec:analysis}

\parab{Startup Phases and Components.} Taking widely adopted LLM inference serving frameworks, vLLM~\cite{kwon2023vllm} and SGLang~\cite{zheng2024sglang}, as examples, we break down the cold-start time in detail. The entire cold-start process can be broadly divided into three stages: \textit{Process Init}, \textit{Model Init}, and \textit{Optimize}. For Qwen3-30B-A3B, vLLM spends 26/58/23 s in the three stages, while SGLang spends 27/58/56 s, respectively. Specifically, in the Process Init stage, the system primarily imports Python packages, initializes the underlying device environment (such as a CUDA context), and sets up inter-process communication (IPC). In the Model Init stage, the system handles the heaviest data transfer tasks, which involve sequentially reading massive model tensors from disk storage to host memory, and subsequently copying them to GPU memory. Finally, in the Optimize stage, the system completes the initialization and allocation of the KV cache, captures execution graphs (such as prefill and decode compute graphs), and compiles the underlying operators, making the model fully ready to process incoming inference requests.

We observe significant similarities in the first two stages, which motivate us to adopt a unified methodology to optimize them together. Specifically, both stages internally initialize multiple identical ``structures'': multiple ``processes'' in Process Init and multiple ``tensors'' in Model Init. We abstract these identical execution entities as ``\textit{components}''. Within the same stage, these components internally undergo the same or highly similar workflows and state transitions.\footnote{The Optimize stage also has room for improvement~\cite{zeng2025medusa}, which is orthogonal and complementary to our method, and can be integrated with our method in the system (\cref{ssec:composability}).}
In the first two stages, there are two kinds of overhead.

$\bullet$ \textbf{The forced sequential execution among components leads to unnecessary serial waiting}. The materialization of processes in Process Init and tensors in Model Init are executed strictly sequentially, following the program's control flow. Viewing components as ``monolithic'' entities and executing them sequentially is highly developer-friendly, as it naturally avoids concurrency risks and facilitates debugging. However, it fails to fully exploit the concurrent processing capabilities of modern servers, and such serialized execution causes the overall latency to accumulate in proportion to the number of components.

$\bullet$ \textbf{The granularity of components mismatches the efficient operational granularity of the underlying hardware.} In Model Init, the system loads tensors from disk to memory and to the GPU one by one, while the underlying storage system prefers block-wise I/O. The core objects in high-level LLM inference algorithms are tensors; the loading programs naturally conform to this algorithmic semantics, which is often not aligned with hardware-friendly access patterns. Tensor sizes are irregular and determined by the model structure, and are often poorly aligned with the preferred I/O granularity of storage devices and interconnects; under tensor parallelism, the dominant parallelization strategy in multi-GPU deployments~\cite{vllm-parallelism-and-scaling}, tensors are partitioned across devices, making transfers less contiguous and more fragmented. The underutilization of hardware bandwidth constitutes a significant bottleneck during cold start.

\subsection{Goal and Challenges}  \label{ssec:challenges}

\parab{Goal and Intuition.} Our goal is to refactor the LLM cold-start program to accelerate the overall cold-start process. Based on the overhead analysis, our optimization intuition is twofold: increasing system concurrency to overlap component execution times, and merging fine-grained operations to fully exploit hardware bandwidth. According to the analysis in \cref{sec:introduction}, directly applying these two ideas faces three challenges.

\parab{Challenge 1: Lack of abstraction for refining internal component logic.} 
The large number and heterogeneity of components prevent direct program modification. We introduce the CFA abstraction to describe component state transitions and dependencies, thereby revealing optimization opportunities for concurrent execution and I/O merging (\cref{ssec:model}).

\parab{Challenge 2: Lack of programming interfaces for applying the CFA model.} 
Components' diverse implementations make it difficult to rewrite them individually into CFA, and their scale also challenges the scalability of the CFA system. We propose a CFA-based programming framework to refactor the original program, which preserves the original program structure as much as possible, enables concurrent execution, and enables isolated CFA state space management (\cref{ssec:programming}).

\parab{Challenge 3: Program correctness risks introduced by concurrent execution.} Substantially refactoring a sequential program into concurrent execution inherently carries a risk of data hazards or even program crashes caused by incorrect execution orders. We provide a rigorous proof demonstrating that concurrent execution based on the CFA model is equivalent to monolithic sequential execution, and analyze data safety in each CFA-refined case (\cref{ssec:correctness}).

\section{The CFA Abstraction} \label{sec:cfa-abstraction}

We use the CFA abstraction to guide the analysis of cross-component optimization and program refactoring, and to prove the correctness of the refactoring.

\subsection{System Model} \label{ssec:model}
\parab{Automata.}
We formalize the LLM cold-start process as a \textit{Communicating Finite Automata (CFA)} system that converges to terminal states. Each physical or logical component involved in the cold start (such as processes, data chunks, etc.) is defined as a \textit{finite automaton (FA)} $c$. A component $c$ is associated with a monotonic state variable $S(c)$, whose value ranges over a finite and ordered state space $\mathcal{S}_c = \{s_{c,0}, s_{c,1}, \dots,$ $s_{c,m_c}\}$. The specific states of the components are defined by the application semantics, but their state transitions throughout the lifecycle must be ``monotonic'' and ``irreversible''. At runtime, each state transition represents the execution of a code block.

\parab{Dependency.} A dependency between components is uniformly expressed as a state-dependency relation, denoted as $(c_i, s_i) \to (c_j, s_j)$. Note that the actual meaning of this dependency is: the ``state transition'' of component $c_j$ from state $s_j$ to state $s_{j+1}$ depends on component $c_i$ having already reached state $s_i$. In other words, only when $c_i$ successfully arrives at state $s_i$ can it trigger the specific operations of component $c_j$ from state $s_j$ to $s_{j+1}$. Because the state transition of each component is monotonic in cold start, for the sake of conciseness, we simplify the state transition as the starting state $(c_j, s_j)$. 

\parab{Runtime Workflow.} In a valid cold start, all components' states and transitions form a \textit{directed acyclic graph (DAG)}. Each component runs the code block of a transition, and publishes the new state on completion; for a dependency, the dependent component waits for the prerequisite condition to be met and proceeds to the next transition; when one component publishes a new state, it may satisfy prerequisite conditions for other components and trigger their transitions (execution).

\begin{figure}[t]
    \begin{lstlisting}[language=Python,numbers=left,
                        numbersep=4pt,firstnumber=1,xleftmargin=1.5em]
    waiters = {}  # dict[tuple[component, state], list[handle]]
    while True:
        event = recv_event()
        key = (event.component, event.state)
        if event.kind == "set_state":
            handles = waiters.get(key, [])
            for h in handles:
                wakeup(h)
            waiters[key] = []
        else:  # wait_state
            if key not in waiters:
                waiters[key] = [event.handle]
            elif len(waiters[key]) > 0:
                waiters[key].append(event.handle)
            else:
                wakeup(event.handle)
    \end{lstlisting}
    
    \caption{Core event loop of a \texttt{channel}.}
    \label{fig:channel-event-loop}
    \end{figure}

\parab{Opportunities from CFA Abstraction.} The CFA abstraction can guide the safe refactoring of LLM cold-start programs targeting the overhead reported in \cref{ssec:problems}. First, the execution processes within FAs are described as state transitions at a finer time granularity. State transitions with no cross-FA dependencies can safely run concurrently. Therefore, the corresponding independent program parts can be refactored, effectively overlapping the execution time of components.

Second, multiple related FAs can be safely merged and simplified, reducing state transitions while maintaining the same semantics from the start to the end state. Specifically, assuming multiple FAs $A_i$ each contain an independent state transition $a^{(i)}_1 \to a^{(i)}_2$, we can logically merge them into a single FA with one state transition $(a^{(1)}_1, a^{(2)}_1, \dots) \to (a^{(1)}_2, a^{(2)}_2, \dots)$. In the program, if the original fine-grained state transitions trigger excessive fragmented I/O requests that mismatch the hardware's ideal I/O granularity, we can safely merge and reorganize these state transitions and their associated I/O operations using FA merging. This can align the refactored I/O granularity with hardware characteristics, thereby achieving higher bandwidth saturation.

\subsection{Programming Framework and Runtime}\label{ssec:programming}

We provide a programming framework to refactor the LLM cold-start program with the CFA abstraction. Components in the program run at different levels (e.g., processes, models, tensors); it is not necessary to manage all FAs together. The CFA framework enables users to specify a set of related FAs in a single namespace, abstracted as a \texttt{channel}. Each \texttt{channel} independently maintains the state space of all components within it, providing a lightweight interaction and synchronization scope; different \texttt{channels} isolate their state management for heterogeneous components.

Developers customize the cold-start program to enable the CFA abstraction. Developers can explicitly express state transitions by inserting the \texttt{channel.set\_state()} primitive directly into the components' existing initialization code. Once a component completes a substantial portion of its working logic, calling this primitive publishes the newly reached monotonic state to the \texttt{channel}, which then wakes all components waiting on that state.

The \texttt{channel.wait\_state()} primitive is inserted into the program's execution flow to declare explicit component dependencies. When the current component reaches this code position, it either blocks until the required state of the other component becomes available or continues immediately if that state has already been published. 

\begin{figure}[t]
    \centering
\begin{lstlisting}[language=Python,xleftmargin=2.5em]
def process_A(channel): # Component A
    A_init1()
    channel.set_state("A", "A1")
    A_init2()
    channel.set_state("A", "A2")
def process_B(channel): # Component B
    B_init1()
    channel.set_state("B", "B1")
    channel.wait_state("A", "A1")
    B_init2()
    channel.set_state("B", "B2")
\end{lstlisting}

    \caption{Process/Thread Examples of CFA Programming}
    \label{fig:code_examples}
\end{figure}

In implementation, the \texttt{channel} runs continuously in the background as a daemon process (or daemon thread) responsible for handling all \texttt{set\_state} and \texttt{wait\_state} events. Figure~\ref{fig:channel-event-loop} illustrates its internal workflow.
This daemon process/thread maintains, for each published state key (\texttt{component}, \texttt{state}), a queue of blocked waiters stored in a dictionary (line 1). Upon receiving a \texttt{set\_state} event (lines 6--9), it looks up the corresponding queue in the dictionary (line 6), wakes all blocked waiters (lines 7--8), and then marks that state as already published by resetting the entry to the empty list (line 9). Upon receiving a \texttt{wait\_state} event (lines 11--16), it either creates (line 12) or appends (line 14) to the corresponding queue, or resumes the waiter immediately if the required state has already been published (line 16).

\parab{An Example.} The CFA programming framework can be applied to concurrent processes, threads, and coroutines. In the example shown in \cref{fig:code_examples} (coroutine example in \cref{app:fig:code_examples} in \cref{sec:coroutine}), we consider two components, $A$ and $B$, with states $\{A1, A2\}$ and $\{B1, B2\}$, and a dependency $(A, A1)$ $\to$ $(B, B1)$; in the cases of synchronous processes/threads and asynchronous coroutines, developers can use identical interfaces to achieve concurrent execution and safe synchronization on dependencies, and keep most of the original code unchanged.

\subsection{Correctness Analysis} \label{ssec:correctness}
We formulate the CFA-refined program and the original program as DAGs, and prove that both programs terminate and converge to the same final states. The definition and theorems are listed below, with the proofs in \cref{sec:correctness}.

\begin{definition}[DAG CFA]
If all internal state transitions and cross-component state dependencies of a CFA form a DAG, we define it as a DAG CFA.
\end{definition}

\begin{definition}[Chain CFA]
Traditional monolithic sequential component execution, whose execution order is strictly governed by the program's control flow, can be abstracted into a special DAG CFA, termed a Chain CFA. 
\end{definition}

\begin{theorem}\label{thm:termination}
A DAG CFA is guaranteed to terminate.
\end{theorem}

\begin{theorem} \label{thm:equivalence}
If the internal state transition sequence of each FA in a DAG CFA is identical to that in a Chain CFA, then when both CFAs terminate, the final states of all FAs will be exactly the same.
\end{theorem}

\begin{figure}[t]
    \centering
    \includegraphics[width=\columnwidth]{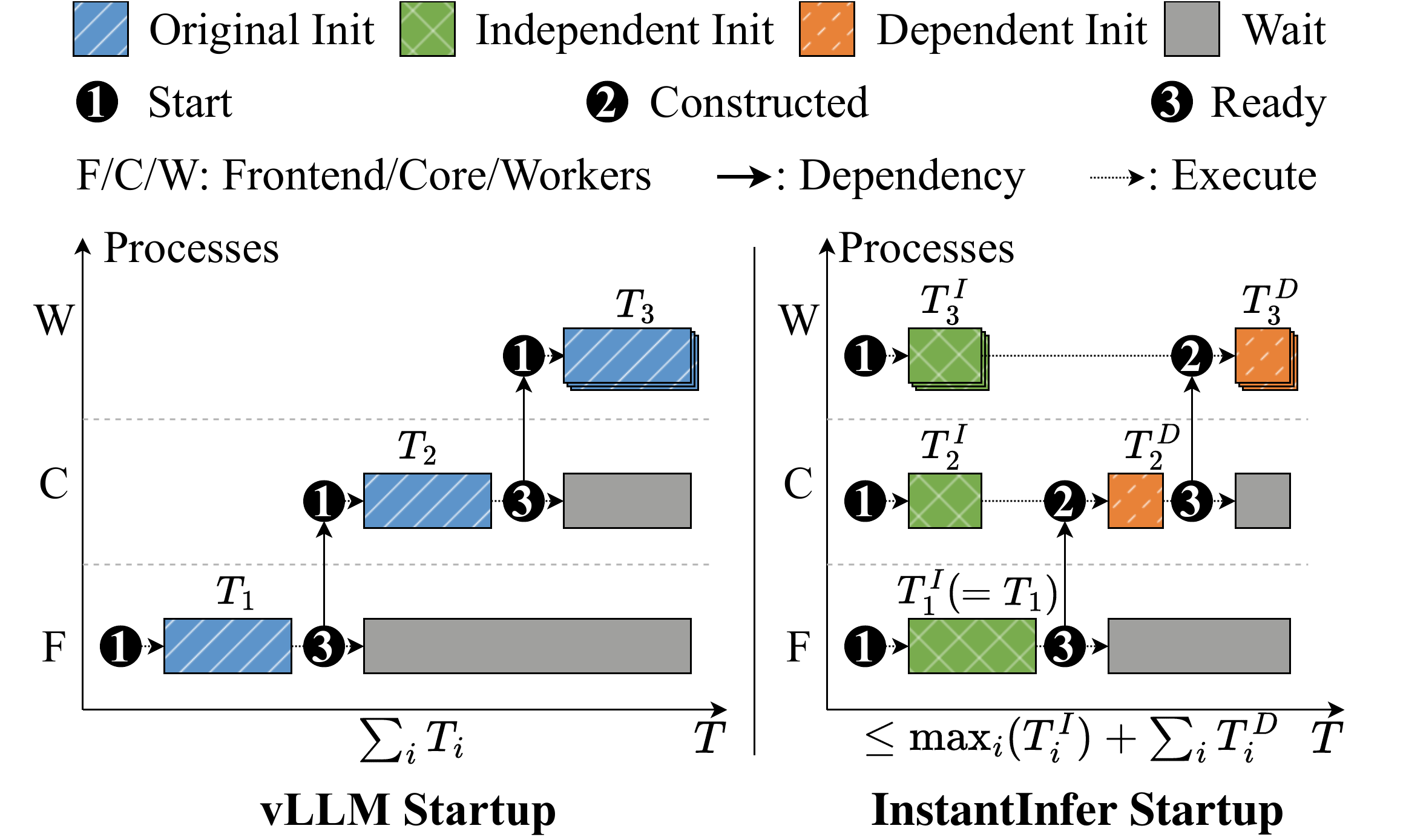}
    \caption{Process-tree materialization in conventional engines and under CFA-guided execution.}
    \label{fig:process-startup}
\end{figure}

\section{CFA for LLM Cold Start} \label{sec:cfa-for-llm}
We apply the CFA abstraction and programming framework to three key procedures of LLM cold start.

\subsection{Process-Tree Materialization} \label{ssec:process-materialization}

Modern LLM inference systems organize their processes in a three-level hierarchy with one frontend process, one core process, and a set of worker processes, where each worker is bound to one GPU. The traditional initialization procedure strictly follows a serial execution protocol: a parent process must complete its own initialization entirely before spawning and configuring its child processes. This level-by-level progression makes the cold-start time proportional to the depth of the process tree. We observe that these processes share similar execution logic and monotonic progress during cold start, and therefore abstract them into a CFA.

\parab{CFA Model.} We define the frontend, core, and workers in the system as FAs in the same channel, where an FA maintains the states of \texttt{\{Start, Constructed, Ready\}} (\ding{182}, \ding{183}, and \ding{184} in \cref{fig:process-startup}). The transition from {\tt Start} to {\tt Constructed} usually involves importing packages, establishing IPC connections, and initializing devices (like CUDA contexts).

In the traditional serial startup logic, the dependency between a parent process and a child process is coarse-grained: the parent must first reach \texttt{Ready}, and the child is launched as \texttt{Start}, i.e., $(\texttt{parent}, \texttt{Ready})$ $\to$ $(\texttt{child},\texttt{Start})$.

Refining each component lifecycle into three states exposes an intermediate dependency boundary: a child can complete the independent stage of local setup (from \texttt{Start} to \texttt{Constructed}) before the parent is \texttt{Ready}, and blocks only at the final parent-dependent stage, i.e., $(\text{\tt parent}, \text{\tt Ready})$ $\to$ $(\text{\tt child}, \text{\tt Constructed})$. 

\parab{Runtime Workflow.} \cref{fig:process-startup} (left) shows the workflow of the traditional monolithic sequential process initialization. The overall time is the sum of the times of Frontend, Core, and (the maximum of) Workers.

The CFA-refined workflow is as follows: first, the system concurrently creates all processes. Each process independently performs initialization operations, such as importing packages and allocating local memory; upon completion, each process independently enters the \texttt{Constructed} state. 

The Frontend process then advances to the \texttt{Ready} state and waits for completion signals from other processes. Meanwhile, after the Core process evolves to \texttt{Constructed} on its own, it suspends and waits for the Frontend process to reach the \texttt{Ready} state; only when this condition is met can the Core process obtain the necessary context and enter the parent-dependent stage. 

Similarly, after the Workers complete the independent initialization and enter the {\tt Constructed} state, they wait for the Core process to reach the \texttt{Ready} state to complete their parent-dependent stage.

\parab{Time Reduction.} In the traditional vLLM sequential initialization workflow, let $T_1$, $T_2$, and $T_3$ denote the initialization time of the Frontend, Core, and Workers, respectively. The total initialization time is therefore
$T_{\mathrm{vLLM}} = \sum_{i=1}^{3} T_i$.

With \sysname refinement, each process initialization is decomposed into an independent stage and a parent-dependent stage, so that $T_i = T_i^I + T_i^D$, where $T_1^D=0$ because the frontend has no parent-dependent stage. The total initialization time becomes the maximum time over all root-to-leaf paths:\\
\centerline{\small
$
T_{\mathrm{CFA}} = \max_i \{ T_i^I + \sum_{j=i}^{3} T_j^D \}
\leq \max_i \{T_i^I\} + \sum_{i=1}^{3} T_i^D
< T_{\mathrm{vLLM}}.
$
}

The first inequality holds because the independent stage is overlapped across processes, while the parent-dependent stage remains ordered along the parent--child chain. The second strict inequality holds as long as at least two processes have nonzero independent work, which is the common case in practice. Therefore, the refined initialization time is strictly smaller than the original serial initialization time.

\parab{Correctness and Safety.} By theorems in \cref{ssec:correctness}, the CFA-refined program terminates with the same final states as the original program. If a child state depends on a parent state, for example, when the child must wait for the parent to prepare an IPC address before establishing the connection, the parent state needs to be initialized during the transition \texttt{Constructed} $\to$ \texttt{Ready}, and the child state can be allocated before 
\texttt{Constructed} but must be initialized after \texttt{Constructed}.

\subsection{Tensor Materialization} \label{ssec:tensor-materialization}

In the traditional LLM cold start, the system materializes tensors from the model file one by one, sequentially reading them from disk to host memory, and then copying them from host memory to GPU memory. We observe that the massive tensors share highly similar movement patterns, so we naturally abstract them as a CFA.

\parab{Na\"ive CFA Model.} We first propose a na\"ive CFA model to refine the workflow. On a single distributed worker node (rank), we define all model tensors (including sharded subsets) as FAs, and refactor the sequential loading workflow. 
In the na\"ive CFA model, each tensor component has three monotonically increasing states: \{\texttt{InDisk}, \texttt{InMem}, \texttt{InGPU}\}. 

Since the tensors are logically independent of each other, theoretically, there is no need to enforce strict CFA state dependencies among them. However, allowing all tensors to perform I/O operations concurrently without restriction can easily cause underlying bus congestion, paradoxically preventing I/O bandwidth saturation. 

\begin{figure}[t]
    \centering
    \includegraphics[width=\linewidth]{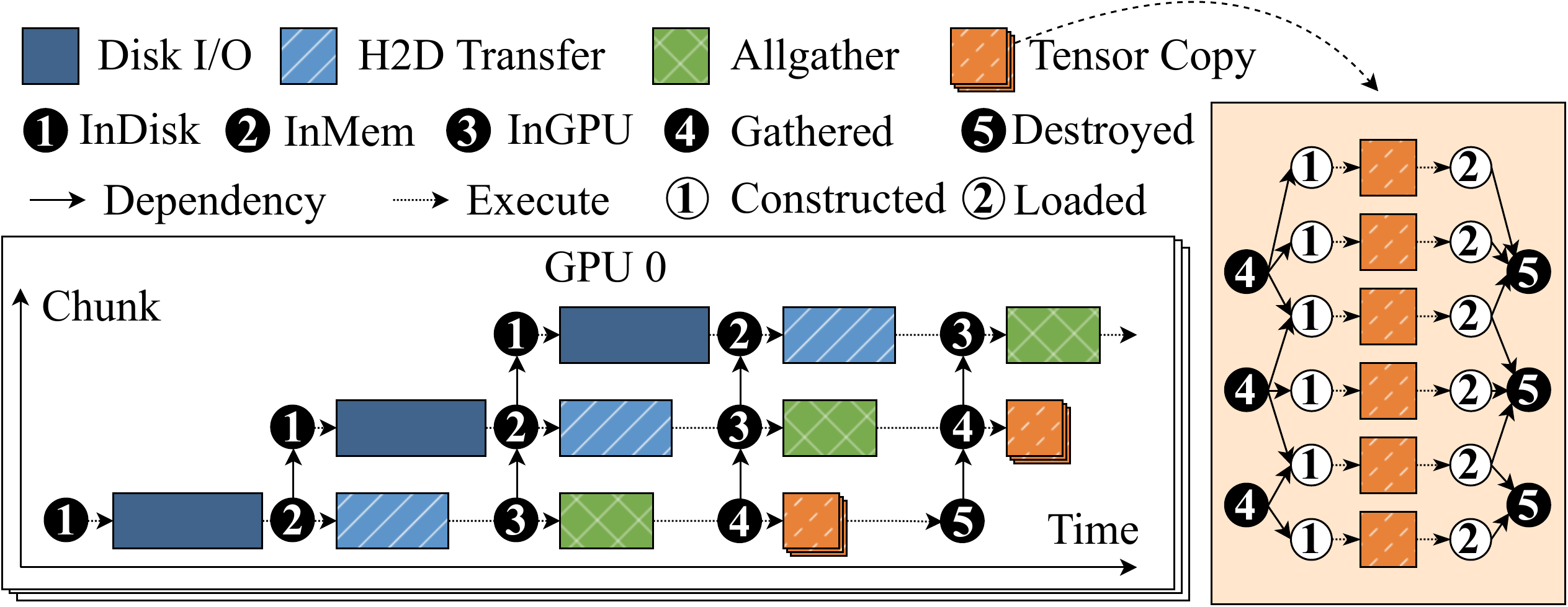}
    \caption{CFA-guided model loading.}
    \label{fig:model-loading-cfa}
\end{figure}

To limit the contention, we explicitly set state dependencies $(T_i, \texttt{InMem}) \to (T_{i+1}, \texttt{InDisk})$ on the disk-to-memory data path to enforce sequential tensor copying along this path; similarly, on the memory-to-GPU path, we set a comparable dependency 
$(T_i, \texttt{InGPU})$ $\to$ $(T_{i+1}, \texttt{InMem})$. 

At runtime, this forms a two-stage pipeline for the disk-memory-GPU data paths. Each pipeline stage loads only one tensor at a time, while the two stages enable efficient concurrency and temporal overlap.

\parab{Refined CFA Model.} Although the na\"ive model introduces a pipeline, performing multiple independent I/O calls for each fine-grained tensor on both paths separately incurs excessive system overhead, which is non-negligible for large models. To further optimize the loading process, we must aggregate operations. However, logically consecutive tensors in the algorithm description are not necessarily adjacent in the actual disk storage files, making it impossible to merge the I/O operations for adjacent tensors directly. 

To this end, we introduce a novel execution component into the system: the data chunk (\texttt{chunk}). We divide the entire physical model file on disk into multiple \texttt{chunks} of equal length (the last \texttt{chunk} may be partial), and precompute the mapping associations between \texttt{chunks} and individual logical tensors. If a tensor's content falls within the range of a certain \texttt{chunk}, the system associates the tensor with that \texttt{chunk}. One tensor may span and be associated with one or multiple \texttt{chunks}, and conversely, one \texttt{chunk} may be associated with one or multiple tensors. We represent this association using \texttt{tensor.chunks} and \texttt{chunk.tensors}, and the system precomputes this mapping in advance. 

The states of the \texttt{chunk} component are \{\texttt{InDisk}, \texttt{InMem}, \texttt{InGPU}, \texttt{Destroyed}\} (\ding{182}, \ding{183}, \ding{184}, and \ding{186} in \cref{fig:model-loading-cfa}); meanwhile, in the refined CFA model, we redefine each tensor component to have two states: \{\texttt{Alloc}, \texttt{Loaded}\} (\ding{172} and \ding{173} in \cref{fig:model-loading-cfa}).

The cross-component dependencies are designed for efficient data flow (\cref{fig:model-loading-cfa}). \textbf{(i)} Among the \texttt{chunk}s, we set dependencies $(C_i, \texttt{InMem})$ $\to$ $(C_{i+1}, \texttt{InDisk})$ and $(C_i, \texttt{InGPU})$ $\to$ $(C_{i+1}, \texttt{InMem})$. These dependencies form a two-stage pipeline with temporal overlap (similar to the na\"ive model). 
\textbf{(ii)} Between a \texttt{chunk} and its associated tensors, we establish the dependency $(C_i, \texttt{InGPU}) \to (T_j, \texttt{Alloc})$. This means that when a required \texttt{chunk} is ready in the GPU memory, it can trigger the associated tensor to transition from \texttt{Alloc} to \texttt{Loaded}, performing efficient data copying and assembly within the GPU to make the tensor ready for use.
\textbf{(iii)} Between a tensor and its associated \texttt{chunk}s, we set the dependency $(T_i, \texttt{Loaded}) \to (C_j, \texttt{InGPU})$. These dependencies mean that once all tensors associated with a chunk have reached \texttt{Loaded}, the chunk can transition from \texttt{InGPU} to \texttt{Destroyed}, releasing GPU memory.

\parab{Optimization of Distributed Loading.}
To further improve hardware utilization, we introduce an optimization mechanism for the multi-path concurrent loading of \texttt{chunk}s. If the system is configured with multiple ranks (e.g., multiple GPUs), we load each \texttt{chunk} cooperatively across multiple paths: each \texttt{chunk} is divided into multiple equal-length ``segments'' according to the number of ranks. 

Within each rank, the assigned segment is loaded following the aforementioned pipeline process. Once these segments become locally ready in their respective GPUs, the system performs an additional \texttt{AllGather} collective communication operation, after which the corresponding \texttt{chunk} transitions to an additional state, \texttt{Gathered}, which triggers the chunk-to-tensor data copy. By stitching these segments together via cross-device transmission, each rank ultimately acquires the complete \texttt{chunk} data. 

This optimization fully utilizes data-transfer channels between ranks and disk for concurrent reads, and then uses the peer-to-peer network among ranks to assemble the complete data.

\parab{Correctness and Safety.} By theorems in \cref{ssec:correctness}, the CFA-refined program converges to the same state as the original program. For data safety, each logical tensor is assigned values after its own allocation and its associated chunks' readiness, which guarantees the tensor values are valid.

\subsection{Runtime Model Switching} \label{ssec:model-switching}

In traditional dynamic multi-model serving systems, model switching is typically treated as an indivisible serial phase transition: the system first shuts down and unloads the current old model, reclaims its associated GPU memory and system resources, and then loads and initializes the new model from scratch. This purely serial approach inevitably introduces substantial service disruption time. We observe that the new model's initialization contains an early GPU-independent stage and a later GPU-dependent stage, while the old model's teardown first releases GPU resources and only then reclaims CPU- and host-memory-side resources. We capture this overlap opportunity by abstracting the old and new models uniformly as CFAs, as illustrated in \cref{fig:model-switching-cfa}.

\begin{figure}[t]
    \centering
    \includegraphics[width=\linewidth]{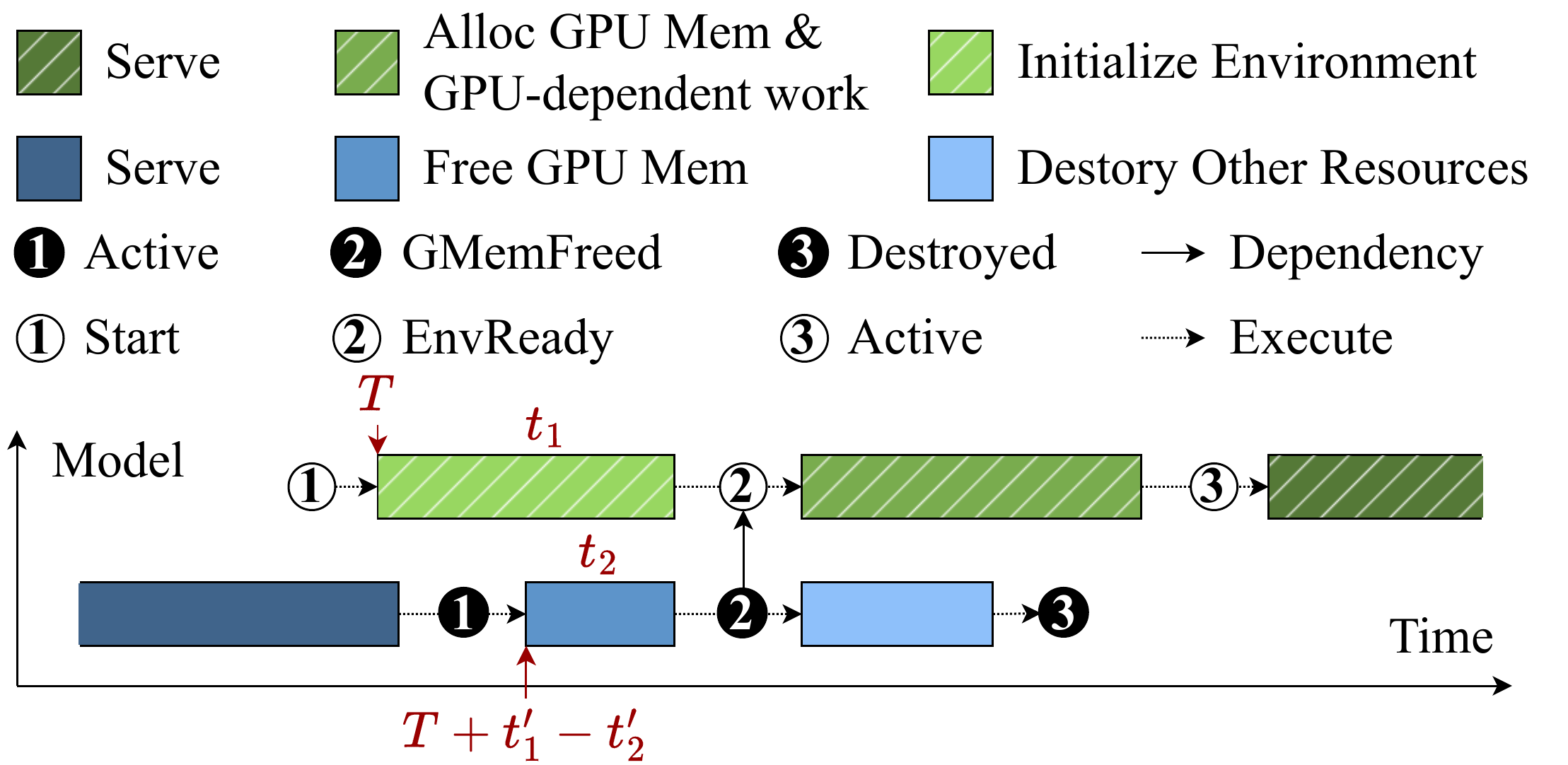}
    \caption{CFA-guided runtime model switching.}
    \label{fig:model-switching-cfa}
\end{figure}

\parab{CFA Model.} In the CFA model for model switching, the old model ($M_1$) and the new model ($M_2$) are treated as two concurrently executing core components. We define the lifecycle states of the old model as: \{\texttt{Active}, \texttt{GMemFreed}, \texttt{Destroyed}\} (\ding{182}, \ding{183}, and \ding{184} in \cref{fig:model-switching-cfa}); and the lifecycle states of the new model as: \{\texttt{Start}, \texttt{EnvReady}, \texttt{Active}\} (\ding{172}, \ding{173}, and \ding{174} in \cref{fig:model-switching-cfa}). Because the old and new models compete for mutually exclusive resources like GPU memory, we explicitly establish a cross-component state dependency between them: $(M_1, \texttt{GMemFreed})$ $\to$ $(M_2, \texttt{EnvReady})$. This dependency means that the new model can trigger the state transition from \texttt{EnvReady} to \texttt{Active} only after the old model has fully released its GPU memory, thereby safely proceeding with GPU-memory allocation operations.

\parab{Workflow with an Optimization.} Running the two models with the CFA refinement above achieves safe concurrent model switching by overlapping execution (\cref{fig:model-switching-cfa}). We also provide an optimization to further reduce the ``new model waiting time''. The main idea is to make the two states ($M_1$, \texttt{GMemFreed}) and ($M_2$, \texttt{EnvReady}) as close as possible, so that the new model does not block for a long time waiting for GPU memory to be released.

Let $T$ be the new model launch time ($M_2$, \texttt{Start}), $t_1$ the new model's environment initialization time ($M_2$, \texttt{Start}) $\to$ ($M_2$, \texttt{EnvReady}), and $t_2$ the old model's GPU-memory release time ($M_1$, \texttt{Active}) $\to$ ($M_1$, \texttt{GMemFreed}). Using historical measurements or offline profiling, we obtain $t'_1$ and $t'_2$ as predictions for $t_1$ and $t_2$. Thus, we set the old model's shutdown time to be $\max\{T + t'_1 - t'_2, T\}$ (\cref{fig:model-switching-cfa}). This schedule minimizes the time the new model waits for GPU memory to be ready.

\section{Implementation} \label{sec:implementation}

We integrate \sysname with vLLM~\cite{kwon2023vllm} by replacing blocking initialization routines with the CFA programming abstraction, which modifies \textasciitilde3000 lines of Python code while preserving the execution semantics of the underlying system. To support the chunk-level I/O pattern efficiently, we further develop a lightweight C++ extension module of \textasciitilde2700 lines for fast chunk I/O across GDS-backed storage (via \texttt{cuFile}), legacy storage (via \texttt{libaio}), and in-memory storage such as tmpfs (via \texttt{cudaMemcpyAsync}). These changes demonstrate that CFA can be incrementally adopted in practice without intrusive system redesign.

During tensor materialization, the AllGather operation is executed over communication groups, which are created by \texttt{torch.distributed.new\_group()} with the NCCL backend.
For example, on four GPUs, users may specify two disjoint two-GPU groups instead of the full set.
By default, we use the world group containing all GPUs.

\section{Evaluation} \label{sec:evaluation}

\begin{figure*}[t]
    \centering

    \includegraphics[width=0.45\textwidth]{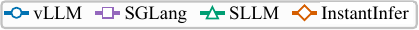}
    \\

    \begin{subfigure}[t]{0.325\textwidth}
        \centering
        \includegraphics[width=\linewidth]{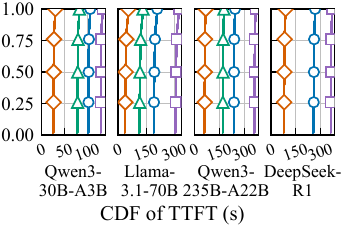}
        \caption{[H20] Single request.}
        \label{fig:eval-single}
    \end{subfigure}
    \hfill
    \begin{subfigure}[t]{0.325\textwidth}
        \centering
        \includegraphics[width=\linewidth]{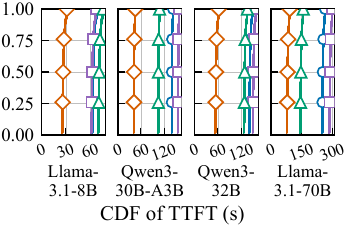}
        \caption{[L40] Single request.}
        \label{fig:eval-single-l40}
    \end{subfigure}
    \hfill
    \begin{subfigure}[t]{0.325\textwidth}
        \centering
        \includegraphics[width=\linewidth]{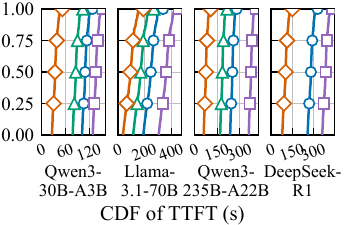}
        \caption{[H20] 32-request burst.}
        \label{fig:eval-burst}
    \end{subfigure}

    \caption{Cold-start TTFT CDFs under different hardware and request settings.}
    \Description{CDF plots comparing TTFT across vLLM, SGLang, SLLM, and InstantLLM under H20 and L40 settings.}
    \label{fig:eval-ttft-cdf}
\end{figure*}

This section addresses five key questions:
(1) How does \sysname affect end-to-end cold-start latency across diverse configurations? (\cref{ssec:eval-end-to-end})
(2) How does cold-start latency scale with the number of instances for \sysname? (\cref{ssec:eval-scalability})
(3) How does \sysname impact model switching latency? (\cref{ssec:eval-switching})
(4) What is the performance contribution of each \sysname mechanism? (\cref{ssec:eval-ablation})
(5) What is the hardware utilization and resource overhead of \sysname? (\cref{ssec:eval-hardware})

\subsection{Evaluation Settings} \label{ssec:settings}

\noindent
\textbf{Testbed.}  
We evaluate \sysname on two representative GPU clusters, each comprising 8 NVIDIA GPUs, to capture both high-end and more commonly deployed configurations. The first testbed uses NVLink-connected H20 (141~GB) GPUs with 50~GB/s networked storage (GPFS~\cite{gpfs}), representing a high-bandwidth, tightly coupled environment. The second employs PCIe-connected L40 (48~GB) GPUs with 20~GB/s networked storage, reflecting a more bandwidth-constrained and cost-efficient deployment. Each node is equipped with 2~TB of host memory and an Intel\textregistered~Xeon\textregistered~Platinum 8468 CPU (192 logical cores). Unless otherwise specified, we load models from storage rather than memory to reflect realistic cold-start scenarios in production. 

\noindent
\textbf{Baselines.}  
We compare \sysname against three representative open-source inference frameworks: vLLM (0.13.0)~\cite{kwon2023vllm}, SGLang (0.5.9)~\cite{zheng2024sglang}, and ServerlessLLM (abbreviated as SLLM in figures; version 0.8.0)~\cite{fu2024serverlessllm}. vLLM and SGLang are among the most widely deployed LLM serving systems, representing state-of-the-art performance in conventional (non-serverless) settings. In contrast, ServerlessLLM is specifically designed for serverless environments and incorporates optimizations for cold-start latency.  All baselines rely on Safetensors~\cite{safetensors} as the de facto standard weight format, while ServerlessLLM further requires converting model weights into a proprietary format for each parallelism configuration before execution. We also prepare and enable the torch compile cache~\cite{torch-compile-cache} for all systems, which is the default and recommended configuration for these systems.

To isolate model loading performance, we further compare against standalone loaders integrated into existing systems, including Safetensors (0.7.0), Run:ai Model Streamer (0.15.6)~\cite{runai-model-streamer}, and fastsafetensors (0.2.2)~\cite{yoshimura2025speeding}. We also include ServerlessLLM as a point of comparison due to its state-of-the-art, albeit proprietary, loading pipeline.

\noindent
\textbf{Models.}  
We select parallelism conservatively to match practical serving constraints following vLLM's guidance~\cite{vllm-parallelism-and-scaling}, balancing memory footprint and parallel efficiency. For dense models, we choose the minimum tensor parallelism (TP) degree required to fit the model weights while preserving KV-cache capacity on the target GPUs, since larger TP increases cross-GPU communication with limited benefit once memory feasibility is satisfied. For MoE models, we additionally use expert parallelism (EP) to shard experts and reduce expert-side load imbalance.
Under this policy, we consider four models on the H20 testbed: (1) Qwen3-30B-A3B~\cite{yang2025qwen3technicalreport} with TP=1; (2) Llama-3.1-70B~\cite{llama-3.1} with TP=2; (3) Qwen3-235B-A22B~\cite{yang2025qwen3technicalreport} with TP=4 and EP=4; and (4) DeepSeek-R1~\cite{guo2025deepseek} with TP=8 and EP=8. 

On the L40 testbed, we evaluate four additional configurations that better match the reduced GPU-memory capacity: (1) Llama-3.1-8B~\cite{llama-3.1} with TP=1; (2) Qwen3-30B-A3B~\cite{yang2025qwen3technicalreport} with TP=2 and EP=2; (3) Qwen3-32B~\cite{yang2025qwen3technicalreport} with TP=4; and (4) Llama-3.1-70B~\cite{llama-3.1} with TP=8. 

\noindent
\textbf{Workload.} Following standard practices in prior work~\cite{kwon2023vllm,fu2024serverlessllm}, we construct realistic workloads using the ShareGPT~\cite{shareGPT} dataset, which contains ChatGPT conversation histories. We filter the dataset to include only requests with prompt lengths up to 32K tokens, ensuring compatibility with the context limits of all evaluated models.

\noindent
\textbf{Metrics.}  
Our primary end-to-end metric is time to first token (TTFT), which captures user-perceived latency in serverless LLM serving. TTFT includes engine startup latency, queuing delay, and prefilling latency, with engine startup typically dominating the critical path in cold-start scenarios.  

To better understand performance bottlenecks, we additionally report fine-grained metrics in our ablation study, including engine startup time, model loading time, and hardware utilization.

\subsection{End-to-End Cold-Start Latency}
\label{ssec:eval-end-to-end}

\noindent
\textbf{Single-request cold start.} We first evaluate the end-to-end cold-start TTFT for a single request, representing the baseline startup efficiency. For each measurement, we generate a request from the dataset and immediately launch a fresh engine instance to process it, repeating this procedure across multiple trials to capture the latency distribution. As shown in \cref{fig:eval-single}, \sysname achieves the lowest cold-start TTFT across all evaluated models. Compared to ServerlessLLM, the state-of-the-art serverless inference system, \sysname achieves a 2.7$\times$--2.9$\times$ speedup in single-request TTFT, owing to faster process and tensor materialization via CFA. Furthermore, \sysname achieves a 3.2$\times$--7.2$\times$ speedup over standard inference engines like vLLM and SGLang, where SGLang exhibits higher latency than vLLM due to its longer compilation and graph capture overhead.

On the L40 testbed, \sysname maintains a significant performance advantage. As shown in \cref{fig:eval-single-l40}, \sysname achieves a 1.8$\times$--2.6$\times$ speedup in single-request cold start over ServerlessLLM. When compared to standard inference engines, \sysname achieves a 2.1$\times$--3.7$\times$ speedup over vLLM and SGLang.

\begin{figure}[t]
    \centering
    \includegraphics[width=\linewidth]{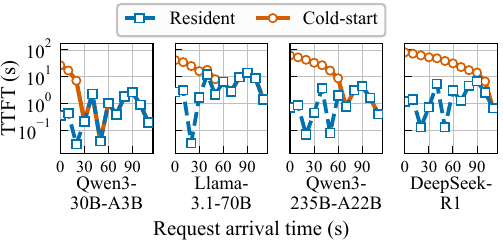}
    \caption{[H20] TTFT for the same requests on cold-start and resident engines.}
    \label{fig:eval-cold-to-hot}
\end{figure}

\begin{figure}[t]
    \centering
    \includegraphics[width=\linewidth]{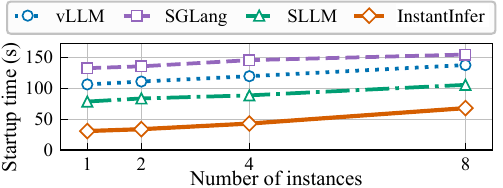}
    \caption{[H20, Qwen3-30B-A3B] Multi-instance startup time with in-memory caching.}
    \label{fig:eval-multi-instance}
\end{figure}

\noindent
\textbf{Burst cold start.} In modern serving environments, cold starts are frequently triggered by sudden traffic spikes. To evaluate this scenario, we extend our methodology by launching a fresh engine instance while dispatching a concurrent burst of sampled requests, repeating this procedure across multiple trials to capture the latency distribution. \cref{fig:eval-burst} shows the TTFT distribution for a burst size of 32, while additional results for burst sizes of 16 and 64 are provided in \cref{sec:burst}. Under these conditions, \sysname maintains its performance advantage, achieving a 1.8$\times$--2.5$\times$ TTFT speedup over ServerlessLLM and a 2.6$\times$--4.9$\times$ speedup over vLLM and SGLang. Compared with the single-request setting, the relative improvement is slightly smaller because a larger burst increases the post-startup portion of TTFT, especially request scheduling and prefill queuing overheads. Across the evaluated settings, increasing the burst size by one raises the TTFT by approximately 0.27\,s on average. These results show that \sysname can still bring the model online quickly enough to absorb substantial traffic bursts.

\noindent
\textbf{Impact of cold starts on steady-state latency.} To verify that \sysname's rapid initialization does not degrade steady-state performance, we replay a sequence of requests sampled from the dataset at fixed 10-second intervals in two settings: a cold-start setting, where a \sysname engine is launched upon the arrival of the first request, and a resident setting with an already-running vLLM engine. As shown in \cref{fig:eval-cold-to-hot}, requests arriving during the initial cold-start phase experience higher latency as the system materializes. However, once the engine is fully initialized and the pending queue is cleared, \sysname's request latency rapidly converges to match that of the resident vLLM engine. This confirms that \sysname accelerates startup without introducing any runtime overhead during continuous inference.

\begin{figure}[t]
    \centering
    \includegraphics[width=\linewidth]{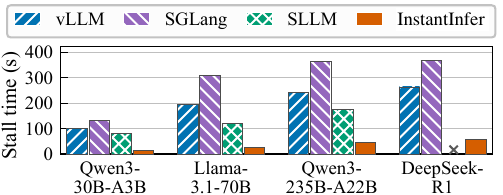}
    \caption{[H20] Service stall time during model switching.}
    \label{fig:eval-switching}
\end{figure}

\begin{figure}[t]
    \centering
    \includegraphics[width=\linewidth]{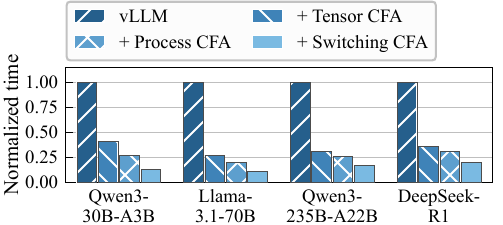}
    \caption{[H20] Startup time improvement from each design.}
    \label{fig:eval-ablation}
\end{figure}

\begin{figure*}[t]
    \centering
    \includegraphics[width=0.6\linewidth]{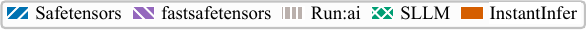}
    \\
    \begin{subfigure}[t]{0.48\textwidth}
        \centering
        \includegraphics[width=\linewidth]{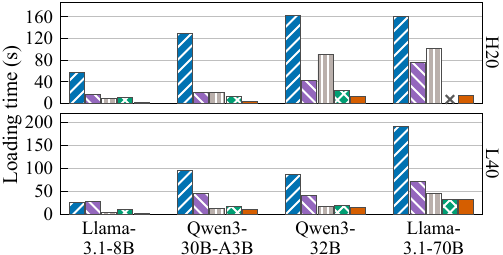}
        \caption{Load from storage.}
    \end{subfigure}
    \hfill
    \begin{subfigure}[t]{0.48\textwidth}
        \centering
        \includegraphics[width=\linewidth]{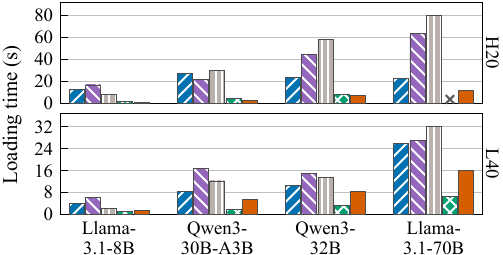}
        \caption{Load from memory.}
    \end{subfigure}
    \caption{[H20 and L40] Model loading time comparison across disk and memory.}
    \label{fig:eval-loading}
\end{figure*}

\subsection{Concurrent Cold-Start Scalability}
\label{ssec:eval-scalability}

When a severe request burst overwhelms a single instance, the serving platform must scale out by concurrently launching multiple instances to maintain quality of service (QoS). We evaluate this scalability using Qwen3-30B-A3B, whose relatively small memory footprint allows the most concurrent instances to be created on a single node, measuring the concurrent startup time for 1, 2, 4, and 8 instances. When multiple instances are launched concurrently, loading the model once into host memory and then copying it to each GPU is substantially more efficient than issuing repeated storage reads for each instance. We therefore enable in-memory model caching, as all evaluated frameworks support this mechanism: ServerlessLLM loads the model into its custom caching layer, vLLM and SGLang load it into the page cache, and \sysname loads it into a tmpfs-based cache.

As shown in \cref{fig:eval-multi-instance}, the relative performance ranking among all frameworks remains consistent with the single-instance scenario, and \sysname achieves the lowest startup time at all concurrency levels. For a single instance, \sysname achieves a 2.6$\times$--4.3$\times$ speedup over the baselines. As the number of concurrent instances scales to 8, \sysname maintains superior efficiency, achieving a 1.6$\times$--2.3$\times$ speedup over the baselines. 
Latency grows with the number of instances, mainly during the process-tree materialization phase, where package imports and CUDA context setup account for most of the increase because concurrent processes contend for I/O, memory allocation, and other system and hardware resources.

\subsection{Model Switching Latency}
\label{ssec:eval-switching}

Dynamic multi-model serving demands low-latency model switching. \cref{fig:eval-switching} quantifies the service stall time---the duration during which no requests can be served---during a model switch. We evaluate switching between models of similar sizes and with the same parallel configurations: Qwen3-32B $\to$ Qwen3-30B-A3B, Qwen2.5-72B $\to$ Llama-3.1-70B, DeepSeek-V2.5 $\to$ Qwen3-235B-A22B, and DeepSeek-V3 $\to$ DeepSeek-R1. Each group is labeled by the target model in the figure.
For the baselines, this stall encompasses sequentially shutting down the outgoing model and cold-starting the incoming one. \sysname reduces this disruption by overlapping the incoming model's initialization with the outgoing model's active serving and shutdown, and by accelerating the cold-start path. It reduces service stall time and achieves a 3.9$\times$--5.0$\times$ speedup over ServerlessLLM and a 4.6$\times$--11.8$\times$ speedup over vLLM and SGLang.

\subsection{Ablation Study}
\label{ssec:eval-ablation}

\noindent
\textbf{Latency reduction by design component.} We isolate the performance gains attributed to each core \sysname mechanism in \cref{fig:eval-ablation}. The \textit{tensor materialization speedup} yields the most substantial gain, achieving a 2.4$\times$--3.7$\times$ speedup over the baseline startup time, as model loading dominates the overall startup latency. Building on this foundation, the speedup of process-tree materialization delivers an additional 1.2$\times$--1.5$\times$ speedup over the optimized baseline. Finally, \textit{multi-model overlapping} removes implicit barriers, contributing a further 1.6$\times$--2.1$\times$ speedup. Altogether, these compounding mechanisms achieve a 5.0$\times$--7.8$\times$ overall speedup compared to the unoptimized baseline.

\begin{figure*}
    \begin{minipage}[t]{0.32\textwidth}
        \centering
        \includegraphics[width=\linewidth]{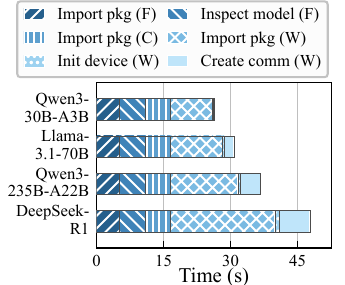}
        \caption{[H20] Breakdown of \underline Frontend/\underline Core/\underline Worker process creation.}
        \label{fig:eval-process-init-breakdown}
    \end{minipage}
    \hfill
    \begin{minipage}[t]{0.268\textwidth}
        \centering
        \includegraphics[width=\linewidth]{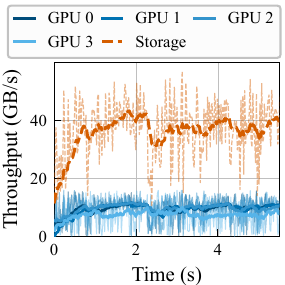}
        \caption{[H20, Qwen3-235B-A22B] GPU/storage throughput.}
        \label{fig:eval-hw-bandwidth}
    \end{minipage}
    \hfill
    \begin{minipage}[t]{0.375\textwidth}
        \centering
        \includegraphics[width=\linewidth]{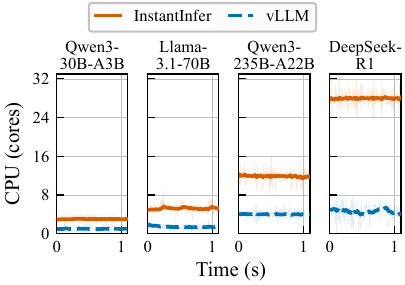}
        \caption{[H20] CPU overhead in model loading.}
        \label{fig:eval-cpu}
    \end{minipage}
\end{figure*}

\noindent
\textbf{Latency reduction of model loading.} \cref{fig:eval-loading} specifically isolates the tensor materialization stage. When loading from storage with 50~GB/s bandwidth on the H20 testbed, \sysname significantly outperforms all alternatives: \sysname achieves a 10.4$\times$--32.3$\times$ speedup over the default Safetensors library, 3.3$\times$--9.0$\times$ over fastsafetensors, 4.8$\times$--7.0$\times$ over Run:ai, and 1.8$\times$--6.5$\times$ over ServerlessLLM. This performance stems from the highly overlapped chunk I/O pipeline constructed through CFA, which fully saturates hardware bandwidth.

When loading from memory on the H20 testbed, \sysname remains highly efficient, achieving a 1.9$\times$--10.1$\times$ speedup over Safetensors and a 1.1$\times$--1.4$\times$ speedup over ServerlessLLM. Compared to loading from storage, Safetensors benefits the most from in-memory loading (4.5$\times$--7.1$\times$ speedup) due to its mmap-based zero-copy pipeline. In contrast, fastsafetensors and Run:ai gain little (1.1$\times$ and 1.2$\times$ on average, respectively), as they are optimized for block storage I/O and incur extra copies and memory allocations that negate the memory bandwidth advantage. ServerlessLLM's loading pipeline also adapts well, achieving a 2.9$\times$--6.3$\times$ speedup through its direct cudaMemcpy-based pipeline. \sysname sees a more modest gain (1.4$\times$--1.8$\times$) because it already gains the highest storage bandwidth when loading from storage, leaving less headroom for in-memory loading to exploit.

\noindent
\textbf{Process-tree materialization time breakdown.} To understand structural bottlenecks, \cref{fig:eval-process-init-breakdown} breaks down the most time-consuming initialization steps. Operations like importing packages, initializing communication groups, and inspecting the model information dominate the critical path. The time for workers to import packages and create communication groups increases with the number of worker processes (which equals the GPU count): the former due to contention for memory and I/O resources among workers, and the latter because more connections and communication buffers need to be created. By using the CFA, \sysname parallelizes these highly synchronous operations, thereby reducing the latency of the critical path.

On the L40 testbed, \sysname achieves 5.7$\times$--8.5$\times$ speedup over Safetensors, 2.2$\times$--8.4$\times$ over fastsafetensors, 1.1$\times$--1.3$\times$ over Run:ai, and 1.0$\times$--3.1$\times$ over ServerlessLLM. This confirms that \sysname adapts robustly to PCIe-bound environments. The narrower gap between \sysname and Run:ai in this setting is attributable to the fact that both systems require inter-GPU communication to forward a subset of weight shards, and the limited PCIe bandwidth becomes the bottleneck for such transfers.
ServerlessLLM does not require such communication, but its transfer overlap is less efficient than \sysname's, thus it performs slightly worse.

When loading from memory on the L40 testbed, \sysname achieves 1.3$\times$--2.9$\times$, 1.7$\times$--4.6$\times$, and 1.6$\times$--2.3$\times$ speedups over Safetensors, fastsafetensors, and Run:ai, respectively. 
ServerlessLLM outperforms \sysname in this setting because its loading pipeline does not require GPU peer-to-peer communication. \sysname, however, still retains two practical advantages: it is compatible with existing model formats and requires less storage space because ServerlessLLM needs multiple copies of weights for different parallelism configurations; moreover, the model-loading logic of \sysname is decoupled from the engine itself and can therefore be integrated into other inference engines, such as SGLang, with little engine-specific modification.

\subsection{Hardware Utilization and Overhead}
\label{ssec:eval-hardware}

We focus on hardware bandwidth utilization to understand why \sysname achieves high loading performance, and on CPU and memory overhead to evaluate whether this design remains friendly to resource-constrained platforms.

\noindent
\textbf{Hardware bandwidth utilization.} 
Model loading dominates engine startup time---e.g., starting a vLLM engine for Qwen3-235B-A22B takes 236s, of which 162s (68.6\%) is spent on loading---so hardware I/O bandwidth utilization directly governs overall cold-start performance. As depicted in \cref{fig:eval-hw-bandwidth}, during the initialization of Qwen3-235B-A22B, \sysname achieves an average storage goodput of nearly 80\% of the theoretical link speed, effectively converting model materialization into a pure hardware-bound transfer with little room for further improvement without faster hardware.

\noindent
\textbf{CPU Overhead.} 
As shown in \cref{fig:eval-cpu}, \sysname utilizes approximately 2.5--3.5 CPU cores per GPU during the intensive model loading stage, compared to roughly 1 core per GPU for vLLM. 
This overhead stems from \sysname's CFA-driven multi-stage I/O pipeline, whose overlapped transfers require additional CPU resources for coordination and progress management.
Given that modern inference servers typically feature hundreds of logical cores (e.g., 192 cores in our testbed) that remain largely idle during engine startup, this modest CPU usage is a highly cost-effective trade-off to aggressively drive high-throughput I/O. Furthermore, on platforms with lower-end CPUs, the underlying storage bandwidth is typically also lower; in such scenarios, \sysname naturally scales down its CPU footprint, requiring fewer cores to fully saturate the available hardware.

\noindent
\textbf{Memory Footprint.} \sysname allocates temporary I/O buffers on both CPU and GPU, sized according to tensor dimensions, storage types, and GPU count to maximize throughput without triggering out-of-memory errors. For instance, loading Qwen3-30B-A3B on a single H20 GPU requires a 4~GB CPU buffer and a 4~GB GPU buffer. For DeepSeek-R1 across eight H20 GPUs, it uses a 4~GB CPU buffer per node and a 512~MB GPU buffer per GPU. Importantly, all allocated buffers are immediately freed once loading completes, ensuring they do not consume space needed for the Key-Value (KV) cache and consequently impose zero memory or performance overhead during steady-state inference.

\section{Related Work}

\noindent
\textbf{Cold-start reduction.}
Recent surveys summarize the rapidly evolving LLM serving landscape
from algorithms to systems~\cite{li2024llminferenceservingsurvey,Miao_2025}.
Within this space, recent systems reduce cold starts through faster
loading, checkpoint locality, pipelining, caching, or state reuse,
including ServerlessLLM~\cite{fu2024serverlessllm},
HydraServe~\cite{lou2025hydraserveminimizingcoldstart},
$\lambda$Scale~\cite{yu2025lambda},
PipeBoost~\cite{liu2025pipeboost},
BlitzScale~\cite{zhang2025blitzscale},
Tangram~\cite{zhu2025tangram},
Medusa~\cite{zeng2025medusa}, and
Foundry~\cite{liu2026foundrytemplatebasedcudagraph},
as well as loading-oriented systems such as InstaInfer,
fastsafetensors, and Run:ai Model Streamer~\cite{sui2024pre, yoshimura2025speeding, runai-model-streamer}.
In contrast, \sysname provides a unifying CFA abstraction that
restructures dependencies across startup, loading, and switching.

\noindent
\textbf{Multi-model serving and switching.}
Multi-model serving and fast switching are studied in
Prism~\cite{yu2025prism},
MuxServe~\cite{duan2024muxserve},
Aegaeon~\cite{xiang2025aegaeon},
Torpor~\cite{yu2025torporgpuenabledserverlesscomputing},
Tangram~\cite{zhu2025tangram}, and
WarmServe~\cite{lou2025warmserveenablingoneformanygpu}, as well as adapter-serving systems
~\cite{chen2024punica, sheng2023s, li2024caraserve, li2025toppings, iliakopoulou2025chameleon, sui2025serverlesslora, ni2025predictive}.
These works optimize sharing, placement, or switching policies,
whereas \sysname models switching as overlapping state progress under
explicit resource constraints.

\noindent
\textbf{Elasticity and scale-out.}
Fast elastic serving depends on burst handling, resource
fragmentation, and cross-cluster scheduling.
Recent systems address these issues with network-assisted scale-out,
cooperative execution, heterogeneous provisioning, and geo-distributed
placement~\cite{yu2025lambda, liu2025pipeboost, zhang2025blitzscale, hu2025deepserve, lv2025dilu, chen2024harmonybatch, mao2025skyserve}.
The mechanisms are complementary to \sysname: by shortening startup
and switching on the critical path, \sysname can improve the
effectiveness of higher-level elastic policies.

\section{Correctness of Program Refactoring}\label{sec:correctness}

We make the following assumption in the correctness proof.
If an FA is executing a transition from state $s_i$ to $s_{i+1}$, it will successfully reach state $s_{i+1}$ within a finite amount of time.

\setcounter{definition}{0}
\setcounter{theorem}{0}

\begin{definition}[DAG CFA]
If all internal state transitions and cross-component state dependencies of a CFA form a DAG, we define it as a DAG CFA.
\end{definition}

\begin{definition}[Chain CFA]
Traditional monolithic sequential component execution, whose execution order is strictly governed by the program's control flow, can be abstracted into a special DAG CFA, termed a Chain CFA. 
\end{definition}

The construction is as follows: If the program's control flow dictates that component $c_i$ must execute strictly before component $c_j$, we explicitly construct a directed dependency edge from the final state of $c_i$ to the initial state of $c_j$. Because such a strict linear control flow will never introduce any cyclic dependencies, the dependency structure of the constructed Chain CFA is inherently acyclic, making it a DAG CFA as well.

\begin{theorem}\label{thm:termination-2}
A DAG CFA is guaranteed to terminate.
\end{theorem}
\begin{proof}
(1) If there are incomplete transitions in the DAG, there must exist at least one node (state) $s$ with an in-degree of $0$ and an out-degree greater than $0$ (denoted as ``in=0$\wedge$out>0''). An in-degree of $0$ indicates that all prerequisites for this state are fully satisfied. According to the execution semantics, the FA to which this state belongs is actively executing its transition to the next state.

(2) By the assumption, after a finite amount of time, the FA will successfully execute to the next state. Once this transition is finished, we remove all directed edges originating from $s$ (including both internal FA state transitions and cross-FA dependencies) from the DAG.

(3) The resulting graph after removing the node's outgoing edges remains a DAG. Therefore, we can repeatedly execute steps 1 and 2. Because the total number of edges in the system is finite, eventually, there will be no nodes left with in=0$\wedge$out>0. Since the DAG is acyclic, all remaining nodes  have both an in-degree and out-degree of $0$, meaning there are no executable steps (state transitions) left. Thus, the system successfully converges and terminates.

\end{proof}

\begin{theorem} \label{thm:equivalence-2}
If the internal state transition sequence of each FA in a DAG CFA is identical to that in a Chain CFA, then when both CFAs terminate, the final states of all FAs will be exactly the same.
\end{theorem}
\begin{proof}
(1) We first prove that when any DAG CFA (including a Chain CFA) terminates, every FA must be exactly at its final defined state. Suppose there is an FA that is not at its final state; then its current state must have a path to advance, meaning its out-degree in the DAG is not $0$. This implies that uneliminated directed edges still exist in the DAG, meaning the DAG must still contain a node with in=0$\wedge$out>0, allowing execution to continue. This directly contradicts the premise that the system has reached a terminated state.

(2) In both CFAs, the partial order of states for each FA is identical. With (1), the final states of all FAs in both the DAG CFA and the Chain CFA are guaranteed to be identical when they terminate.

\end{proof}

\section{Conclusion}

We propose a CFA abstraction and framework to refactor the LLM cold-start program. CFA safely enables the concurrent execution of heterogeneous components and the merging of fine-grained data operations. The framework preserves sequential program structure.
We refine process-tree materialization, tensor loading, and model switching in vLLM cold start, forming the system \sysname. Experiments demonstrate that \sysname effectively accelerates LLM cold start and is robust across diverse GPUs, workloads, and scales.

\newpage\clearpage

\bibliographystyle{ACM-Reference-Format}
\bibliography{ref}

\newpage\clearpage

\appendix
\crefalias{section}{appendix}
\crefalias{subsection}{appendix}
\crefalias{subsubsection}{appendix}

\begin{figure}[t]
    \centering

\begin{lstlisting}[language=Python,xleftmargin=2.5em]
async def coro_A(channel): # Component A
  await A_init1()
  channel.set_state("A", "A1")
  await A_init2()
  channel.set_state("A", "A2")

async def coro_B(channel): # Component B
  await B_init1()
  channel.set_state("B", "B1")
  await channel.wait_state_async("A", "A1")
  await B_init2()
  channel.set_state("B", "B2")
\end{lstlisting}
    \caption{Coroutine Example of CFA Programming}
    \label{app:fig:code_examples}
\end{figure}

\section{Composability}\label{ssec:composability}

\sysname's CFA abstraction is composable with complementary
optimizations that target other components of cold-start latency.
For example, \sysname does not currently optimize the \emph{Optimize} stage discussed in \cref{ssec:analysis} (e.g., compilation and CUDA graph capture). Medusa~\cite{zeng2025medusa} addresses precisely this stage by serializing captured CUDA graphs and restoring them on subsequent cold starts, thereby eliminating the expensive graph capture overhead. Such orthogonal techniques can be directly incorporated into \sysname to further reduce end-to-end cold-start latency.

Similarly, HydraServe~\cite{lou2025hydraserveminimizingcoldstart} distributes model weights across multiple instances in a pipeline-parallel fashion so that serving can begin before any single instance holds the full model, and then gradually transitions to data parallelism by loading the complete weights on each instance to improve throughput. This mechanism can also be incorporated into \sysname, and the higher loading bandwidth provided by \sysname can further shorten the service startup time of pipeline parallelism.

\section{Coroutine Example of CFA Programming}\label{sec:coroutine}

Figure~\ref{app:fig:code_examples} shows a coroutine example of CFA programming, whose logic is similar to that of process- and thread-based execution.

\section{Additional Burst Cold-Start Results}
\label{sec:burst}

\noindent
\textbf{Burst cold-start TTFT.} Beyond the burst size of 32 shown in the main text, \sysname also maintains clear advantages at burst sizes of 16 (\cref{fig:eval-burst-16}) and 64 (\cref{fig:eval-burst-64}). At a burst size of 16, \sysname achieves a 3.0$\times$--3.6$\times$ speedup over vLLM, a 4.1$\times$--5.3$\times$ speedup over SGLang, and a 2.2$\times$--2.7$\times$ speedup over ServerlessLLM. At a burst size of 64, \sysname achieves a 2.0$\times$--2.9$\times$ speedup over vLLM, a 2.8$\times$--4.3$\times$ speedup over SGLang, and a 1.5$\times$--2.2$\times$ speedup over ServerlessLLM. As burst size increases, the relative gain decreases somewhat because a larger fraction of TTFT is spent after engine startup, but \sysname still substantially shortens the cold-start critical path.

\begin{figure}[t]
  \centering
  \includegraphics[width=0.9\linewidth]{fig/ttft_cdf_legend.pdf}
  \\
  \includegraphics[width=0.75\linewidth]{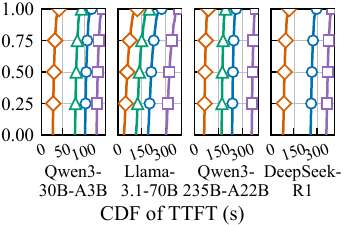}
  \caption{[H20] Cold-start TTFT under a request burst (concurrency=16). ServerlessLLM does not support DeepSeek-R1 because it is integrated with an older engine.}
  \label{fig:eval-burst-16}
\end{figure}

\begin{figure}[t]
  \centering
  \includegraphics[width=0.9\linewidth]{fig/ttft_cdf_legend.pdf}
  \\
  \includegraphics[width=0.75\linewidth]{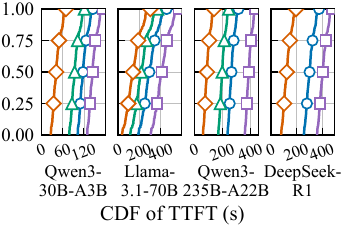}
  \caption{[H20] Cold-start TTFT under a request burst (concurrency=64). ServerlessLLM does not support DeepSeek-R1 because it is integrated with an older engine.}
  \label{fig:eval-burst-64}
\end{figure}

\end{document}